\documentclass[11pt]{article}
\usepackage{amsmath,amssymb,amsfonts,amsthm}
\usepackage{graphicx}
\usepackage{algorithm,algorithmic}
\usepackage{times}
\usepackage{fancyhdr}
\usepackage[small]{titlesec}
\usepackage{color}
\usepackage{enumerate}

\usepackage{hyperref}
\hypersetup{
    linkcolor=red,          
    citecolor=green,        
    filecolor=magenta,      
    urlcolor=cyan           
}

\providecommand{\norm}[1]{\left \lVert#1 \right  \rVert}

\newcommand{\argmax}{\operatornamewithlimits{argmax}}
\newcommand{\argmin}{\operatornamewithlimits{argmin}}

\newcommand{\EE}{\mathbb{E}}           
         
\newcommand{\GG}{\mathbb{G}}
         
\newcommand{\PP}{\mathbb{P}}         
\newcommand{\QQ}{\mathbb{Q}}

\newcommand{\thma}{\begin{thm}}
\newcommand{\thmb}{\end{thm}}

\newcommand{\mata}{\begin{bmatrix}}
\newcommand{\matb}{\end{bmatrix}}

\newtheorem{thm}{Theorem}
\newtheorem{lem}{Lemma}

\newtheorem{defi}{Definition}

\newtheorem{coro}[thm]{Corollary}

\newcommand{\enuma}{\begin{enumerate}}
\newcommand{\enumb}{\end{enumerate}}
\newcommand{\ena}{\begin{enumerate}}
\newcommand{\enb}{\end{enumerate}}

\newcommand{\itema}{\begin{itemize}}
\newcommand{\itemb}{\end{itemize}}
\newcommand{\ita}{\begin{itemize}}
\newcommand{\itb}{\end{itemize}}

\newcommand{\proofa}{\begin{proof}}
\newcommand{\proofb}{\end{proof}}

\newcommand{\bla}{\begin{block}}
\newcommand{\blb}{\end{block}}

\newcommand{\seqb}{\end{equation*}}

\newcommand{\hth}{\mh{\theta}}

\providecommand{\mc}[1]{\mathcal{#1}}
\providecommand{\mb}[1]{\boldsymbol{#1}}

\providecommand{\mh}[1]{\hat{#1}}

\providecommand{\mt}[1]{\widetilde{#1}}

\providecommand{\mtc}[1]{\widetilde{\mathcal{#1}}}

\providecommand{\mht}[1]{\mh{\mt{#1}}}

\newcommand{\conv}{\rightarrow}
\newcommand{\pconv}{\overset{p}{\conv}}

\newcommand{\defa}{\begin{defi}}
\newcommand{\defb}{\end{defi}}

\newcommand{\Qs}{Q}

\title{\vspace{-75pt}Shuffled Graph Classification:\\ Theory and Connectome Applications}

\begin{document}
\maketitle	

\begin{abstract}
	We develop a formalism to address statistical pattern recognition of graph valued data.  Of particular interest is the case of all graphs having the same number of uniquely labeled vertices. When the vertex labels are latent, such graphs are called \emph{shuffled graphs}.  Our formalism  provides insight to trivially answer a number of open statistical questions including: (i) under what conditions does shuffling the vertices degrade classification performance and (ii) do universally consistent graph classifiers exist?   The answers to these questions lead to practical heuristic algorithms with state-of-the-art finite sample performance, in agreement with our theoretical asymptotics.  
\end{abstract}


\section{Introduction} \label{sec:1}

Representing data as graphs is becoming increasingly popular, as technological progress facilitates measuring ``connectedness'' in a variety of domains, including social networks, trade-alliance networks, and brain networks.  While the theory of pattern recognition is deep \cite{Devroye1996}, previous theoretical efforts regarding pattern recognition almost invariably assumed data are collections of vectors.  Here, we assume data are collections of graphs (where each graph is a set of vertices and a set of edges connecting the vertices).  For some data sets, the vertices of the graphs are \emph{labeled}, that is, one can identify the vertex of one graph with a vertex of the others (note that this is a special case of assuming vertices are labeled, where each vertex has a unique label).  For others, the labels are unobserved and/or assumed to not exist.  We investigate the theoretical and practical implications of the absence of vertex labels.  

These implications are especially important in the emerging field of ``connectomics'', the study of connections of the brain \cite{Hagmann05, Sporns2010}.  In connectomics, one represents the brain as a graph (a brain-graph), where vertices correspond to (groups of) neurons and edges correspond to connections between them.  In the lower tiers of the evolutionary hierarchy (e.g., worms and flies), many neurons have been assigned labels \cite{WhiteBrenner86}.  However, for even the simplest vertebrates, vertex labels are mostly unavailable when vertices correspond to neurons.  

Classification of brain-graphs is therefore poised to become increasingly popular.  Although previous work has demonstrated some possible strategies of graph classification in both the labeled \cite{VP11_sigsub} and unlabeled \cite{Duin2011} scenarios, relatively little work has compared the theoretical limitations of the two.  We therefore develop a random graph model amenable to such theoretical investigations.  The theoretical results lead to universally consistent graph classification algorithms, and practical approximations thereof.  We demonstrate that the approximate algorithm has desirable finite sample properties via a real brain-graph classification problem of significant scientific interest: sex classification.

\section{Graph Classification Models} 
\label{sec:shuffler_graph_class_models}

\subsection{A labeled graph classification model} 
\label{sub:a_labeled_graph_classification_model}


A labeled graph $G=(\mc{V},\mc{E})$ consists of a vertex set $\mc{V}$, where $|\mc{V}|=n < \infty$ is the number of vertices, and an edge set $\mc{E}$, where $|\mc{E}| \leq n^2$.
\begin{defi}
Let $\GG\colon \Omega \to \mc{G}_n$ be a \emph{labeled} graph-valued random variable taking values $G\in \mc{G}_n$, where $\mc{G}_n$ is the set of labeled graphs on $n$ vertices.	
\end{defi}
The cardinality of $\mc{G}_n$ is super-exponential in $n$.  For example, when all labeled graphs are assumed to be simple (that is, undirected binary edges without loops), then $|\mc{G}_n|=2^{\binom{n}{2}}=d_n$. 
Let $Y$ be a categorical random variable, $Y\colon \Omega \to \mc{Y}=\{y_0,\ldots, y_{c}\}$, where $c< \infty$.  Assume the existence of a joint distribution, $\PP_{\GG,Y}$ which can be decomposed into the product of a class-conditional distribution (likelihood) $\PP_{\GG|Y}$ and a class prior $\pi_Y$. Because $n$ is finite, the class-conditional  distributions $\PP_{\GG | Y=y}=\PP_{\GG|y}$ can be considered discrete distributions $\text{Discrete}(G; \theta_y)$, where $\theta_y$ is an element of the $d_n$-dimensional unit simplex $\triangle_{d_n}$ (satisfying $\theta_{G|y}\geq 0$ $\forall G \in \mc{G}_n$ and $\sum_{G \in \mc{G}_n} \theta_{G|y}=1$).

\subsection{A shuffled graph classification model} 
\label{sub:a_shuffled_graph_classification_model}


In the above, it was implicitly assumed that the vertex labels were observed.  However, in certain situations (such as the motivating connectomics example presented in Section \ref{sec:1}), this assumption is unwarranted.  To proceed, we define two graphs $G,G' \in \mc{G}_n$ to be isomorphic if and only if there exists a vertex permutation (shuffle) function $\Qs\colon\mc{G}_n \to \mc{G}_n$ such that $\Qs(G)=G'$.  Let $\QQ$ be a permutation-valued random variable, $\QQ\colon \Omega \to \mc{Q}_n$, where $\mc{Q}_n$ is the space of vertex permutation functions on $n$ vertices so that $|\mc{Q}_n|=n!$.  
\begin{defi} \label{def:shuffled}
Let $\GG'=\QQ(\GG): \Omega \to \mc{G}_n$ be a \emph{shuffled} graph-valued random variable, that is, a labeled graph valued random variable that has been passed through a random shuffle channel $\QQ$. 
\end{defi}

Extending the above graph-classification model to include this vertex shuffling distribution yields $\PP_{\QQ,\GG,Y}$.  We assume throughout this work (with loss of generality) that the shuffling distribution is both \emph{class independent} and \emph{graph independent}; therefore, this joint model can be decomposed as
\begin{align}
	\PP_{\QQ,\GG,Y} = \PP_{\QQ} \PP_{\GG,Y} = \PP_{\QQ} \PP_{\GG |Y} \pi_Y = \PP_{\QQ(\GG) |Y} \pi_Y.
\end{align}
As in the labeled case, the shuffled graph class-conditional distributions $\PP_{\QQ(\GG)|y}$ can be represented by discrete distributions $\text{Discrete}(G; \theta_y')$.  Because $\QQ(\GG)$ can be any of $|\mc{G}_n|$ different graphs, it must be that $\theta_y' \in \triangle_{d_n}$.  When $\PP_{\QQ}$ is uniform on $\mc{Q}_n$, all shuffled graphs within the same isomorphism set are equally likely; that is  $\{\theta_{G_i|y}' = \theta_{G_j|y}' \, \forall G_i,G_j \colon \Qs(G_i)=G_j$ for some $\Qs \in \mc{Q}_n\}$.

Note that one can think of a labeled graph as a shuffled graph for which $\QQ$ is a point mass at  $Q=I$, where $I$ is the identity matrix.

\subsection{An unlabeled graph classification model} 
\label{sub:an_unlabeled_graph_classification_model}


The above shuffling view is natural whenever the vertices of the collection of graphs share a set of labels, but the labeling function is unknown.   However, when the vertices of the collection of graphs have different labels, perhaps a different view is more natural.

An \emph{unlabeled graph} $\mt{G}$ is the collection of graphs isomorphic to one another, that is, $\mt{G}=\{Q(G)\}_{Q \in \mc{Q}_n}$. Let $\mt{G}$ be an element of the collection of graph isomorphism sets $\mt{\mc{G}}_n$. The number of unlabeled graphs on $n$ vertices is $|\mt{\mc{G}}_n|=\mt{d}_n \approx d_n/n!$ (see \cite{A000088} and references therein).
An \emph{unlabeling function} $U\colon \mc{G}_n \to \mt{\mc{G}}_n$ is a function that takes as input a graph and outputs the corresponding unlabeled graph. 
\begin{defi}
Let $\mt{\GG}=U(\GG)\colon \Omega \to \mt{\mc{G}}_n$ be an \emph{unlabeled} graph-valued random variable, that is, a labeled graph-valued random variable that has been passed through an unlabeled channel. In other words, $\mt{\GG}=\{Q(\GG)\}_{Q \in \mc{Q}_n}$, and takes values $\mt{G} \in \mt{\mc{G}}_n$. 
\end{defi}  
The joint distribution over unlabeled graphs and classes is therefore
$\PP_{\mt{\GG},Y}=\PP_{U(\GG),Y}=\PP_{U(\QQ(\GG)),Y}$, which decomposes as $\PP_{\mt{\GG}|Y} \pi_Y$. The class-conditional distributions $\PP_{\mt{\GG} | y}$ over isomorphism sets (unlabeled graphs) can also be thought of as discrete distributions $\text{Discrete}(\mt{G}; \mb{\mt{\theta}}_y)$ where $\mb{\mt{\theta}}_y\in \triangle_{\mt{d}_n}$ are vectors in the $\mt{d}_n$-dimensional unit simplex.   Comparing shuffling and unlabeling for the independent and uniform shuffle distribution $\PP_{\QQ}$, we have $\{\theta_{G|y}'=\mt{\theta}_{\mt{G}|y}/|\mt{G}|$ for all $G \in \mt{G}\}$.

\section{Bayes Optimal Graph Classifiers} 
\label{sec:bayes_optimal_graph_classifiers}

We consider graph classification in the three scenarios described above: labeled, shuffled, and unlabeled.  To proceed, in each scenario we define three mathematical objects: (i) a graph classifier, (ii) risk, (iii), the Bayes optimal classifier, and (iv) the Bayes risk.


\subsection{Bayes Optimal Colored Graph Classifiers} 
\label{sub:labeled_graph_classifiers}


	A \emph{labeled graph classifier} $h\colon \mc{G}_n \to \mc{Y}$ is any function that maps from labeled graph space to class space. The risk of a labeled graph classifier $h$ under $0-1$ loss is the expected misclassification rate $L(h)=\EE[h(\GG)\neq Y]$, where the expectation is taken against $\PP_{\GG,Y}$. The \emph{labeled graph Bayes optimal classifier} is given by
	\begin{align}
		h_* = \argmin_{h \in \mc{H}} L(h),
	\end{align}
	where $\mc{H}$ is the set of possible labeled graph classifiers.	The \emph{labeled graph Bayes risk} is given by $L_*=L(h_*)$, 
	where $L_*$ implicitly depends on $\PP_{\GG,Y}$.

\subsection{Bayes Optimal Shuffled Graph Classifiers} 

A \emph{shuffled graph classifier} is also any function $h\colon \mc{G}_n \to \mc{Y}$ (note that the set of shuffled graphs is the same as the set of labeled graphs). However, by virtue of the input being a shuffled graph as opposed to a labeled graph, the shuffled risk under $0-1$ loss is given by $L'(h)=\EE[h(\QQ(\GG)) \neq Y]$, where the expectation is taken against $\PP_{\QQ(\GG),Y}$. 
%
The \emph{shuffled graph Bayes optimal classifier} is given by
\begin{align}
	h_*' &= \argmin_{h \in \mc{H}} L'(h), 
\end{align}
where $\mc{H}$ is again the set of possible labeled (or shuffled) graph classifiers. The \emph{shuffled graph Bayes risk} is given by $L_*'=L(h_*')$,
where  $L_*'$ implicitly depends on $\PP_{\QQ(\GG),Y}$.  

\subsection{Bayes Optimal Unlabeled Graph Classifiers} 

An \emph{unlabeled} graph classifier $\mt{h}\colon \mt{\mc{G}}_n \to \mc{Y}$ is any function that maps from unlabeled graph space to class space. The risk under $0-1$ loss is given by $\mt{L}(\mt{h})=\EE[\mt{h}(\mt{\GG})\neq Y]$, where the expectation is taken against $\PP_{\mt{\GG},Y}$. 
%
The \emph{unlabeled graph Bayes optimal classifier} is given by 
\begin{align}
	\mt{h}_* &= \argmin_{\mt{h} \in \mt{\mc{H}}} L(\mt{h}), \qquad 
\end{align}
%
%
%
The \emph{unlabeled graph Bayes risk} is given by $\mt{L}_*=L(\mt{h}_*)$,
where $\mt{\mc{H}}$ is the set of possible unlabeled graph classifiers
and $\mt{L}_*$ implicitly depends on $\PP_{\mt{\GG},Y}$.  

\subsection{Parametric Graph Classifiers} 
\label{sub:parametric_graph_classifiers}

The three Bayes optimal graph classifiers can be written explicitly in terms of their model parameters:
\begin{align}
	h_*(G) &= \argmax_{y \in \mc{Y}} \theta_{G|y}\pi_y, \label{eq:La_Bayes} \\
	h_*'(G) &= \argmax_{y \in \mc{Y}} \theta_{G|y}' \pi_y,  \label{eq:Sh_Bayes} \\
	\mt{h}_*(\mt{G}) &= \argmax_{y \in \mc{Y}}  \mt{\theta}_{\mt{G}|y}\pi_y. \label{eq:Un_Bayes}
\end{align}


%

\section{Under what conditions does shuffling the vertices degrade classification performance?} 
\label{sec:shuffle}

The result of either shuffling or unlabeling a graph can only degrade, but not improve Bayes risk.  This is a restatement of the data processing lemma for this scenario. Specifically, \cite{Devroye1996} shows that the data processing lemma indicates that in the classification domain $L^*_X \leq L^*_{T(X)}$ for any transformation $T$ and data $X$.  In our setting, this becomes:

\begin{lem} \label{thm:1}
$L_* \leq \mt{L}_*=L_*'$.
\end{lem}

\begin{proof}
	Assume for simplicity $|\mc{Y}|=2$ and $\pi_0=\pi_1=1/2$.  
\begin{align} \label{eq:thm1}
	\mt{L}_*&=\sum_{\mt{G} \in \mt{\mc{G}}_n} \min_y  \mt{\theta}_{\mt{G}|y}  
	=\sum_{\mt{G} \in \mt{\mc{G}}_n} \min_y  \sum_{G \in \mt{G}}\theta_{G|y}'  = L_*' \nonumber \\
	&=\sum_{\mt{G} \in \mt{\mc{G}}_n} \min_y  \sum_{G \in \mt{G}}\theta_{G|y} 
	 \geq \sum_{\mt{G} \in \mt{\mc{G}}_n} \sum_{G \in \mt{G}} \min_y  \theta_{G|y}  = L_*.
\end{align}
\end{proof}

An immediate consequence of the above proof is that the inequality in the statement of Lemma \ref{thm:1} strict whenever the inequality in Eq. \eqref{eq:thm1} is strict:

\begin{lem} \label{thm:2}
	$L_* < \mt{L}_*=L'_*$ if and only if there exists $\mt{G}$ such that
	$$\min_y \mt{\theta}_{\mt{G}|y} > \sum_{G \in \mt{G}}\min_y \theta_{G|y}.$$
\end{lem}
The above result demonstrates that even when the labels \emph{do} carry some class-conditional signal, it may be the case that shuffling or unlabeling does not degrade performance.  In other words, the following two statements are equivalent: (i) the labels contain information with regard to the classification task, and (ii) some graphs within an isomorphism set are class-conditionally more likely than others: $\exists \, \theta_{G_i|y} \neq \theta_{G_j|y}$ where $\Qs(G_i)=G_j$ for some $G_i,G_j \in \mc{G}_n$, $\Qs \in \mc{Q}_n$, and $y \in \mc{Y}$
.  Uniform shuffling has the effect of ``flattening'' likelihoods within isomorphism sets, from $\theta_y$ to $\theta_y'$, so that $\theta_y'$ satisfies $\{\theta_{G|y}'=\mt{\theta}_{\mt{G}|y}/|\mt{G}| \, \forall \colon G \in \mt{G}\}$.  But just because the shuffling changes class-conditional likelihoods does \emph{not} mean that Bayes risk must also change. This result follows immediately upon realizing that posteriors can change without classification performance changing.  The above results are easily extended to consider non-equal class priors and $c$-class classification problems.  To see this, ignoring ties, simply replace each minimum likelihood with a sum over all non-maximum posteriors: 
\begin{align}
\min_y \theta_{G|y} \pi_y \mapsto \sum_{y \in \mc{Y}'} \theta_{G|y} \pi_y  \text{ where } \mc{Y}' =\{y \colon y \neq \argmax_y \theta_{G|y} \pi_y\}.
\end{align}


Prior to concluding this section, we remark that one can achieve Bayes optimal risk using  graph invariants. A graph invariant on $\mc{G}_n$ is any function $\psi$  such that $\psi(G)=\psi(\Qs(G))$ for all $G \in \mc{G}_n$ and $\Qs \in \mc{Q}_n$ (note that an unlabeling function $U(G)$ is a special case of $\psi$).  A graph invariant classifier is a composition of a classifier with an invariant function, $h^\psi=f^\psi \circ \psi$.  The Bayes optimal graph invariant classifier minimizes risk over all invariants: 
\begin{align} \label{eq:psi}
	h_*^{\psi}=\argmin_{\psi \in \Psi, f^{\psi} \in \mc{F}^{\psi}} \EE[f(\psi(\GG))\neq Y],
\end{align}
where $\Psi$ is the space of all possible invariants and $\mc{F}^{\psi}$ is the space of classifiers composable with invariant $\psi$. The expectation in Eq. \eqref{eq:psi} is taken against $\PP_{\GG,Y}$ or equivalently $\PP_{\QQ(\GG),Y}$, since invariants are invariant.
Let $L_*^{\psi}$ denote the Bayes invariant risk.  
\begin{lem} \label{thm:3}
	$\mt{L}_*=L_*^\psi$.
\end{lem}

\begin{proof}
Let $\psi$ indicate in which equivalence set $G$ resides; that is,  $\psi(G)=\mt{G}$ if and only if $G \in \mt{G}$.  Then
\begin{align}
	h^{\psi}_*(G) &= \argmax_{y \in \mc{Y}} \mt{\theta}_{\psi(G)|y} \pi_y \nonumber \\
	&= \argmax_{y \in \mc{Y}} \mt{\theta}_{\mt{G}|y} \pi_y = \mt{h}_*(G).
\end{align}
\end{proof}


\section{Do universally consistent graph classifiers exist?} 
\label{sec:gi}

Throughout this paper, we consider two distinct ``flavors'' of graph classifiers that we can estimate from the data: (i) Bayes plug-in  and (iii) nearest neighbor.  Below, we first introduce Bayes plug-in graph classifiers.  In the following sections, we will discuss their asymptotic properties, as well as the asymptotic properties of $k_s$ nearest neighbor graph classifiers.


\subsection{Bayes Plug-In Graph Classifiers} 
\label{sub:parametric_classifiers}

Implementing the above optimal classifiers requires knowing the model parameters.
When the parameters are unknown (effectively always), we assume that the data are sampled identically and independently from some unknown joint distribution: $(\QQ_i(\GG_i),Y_i)\overset{iid}{\sim}\PP_{\QQ,\GG,Y}$.  
For \emph{labeled} graph classification, $\PP_{\QQ}$ is assumed to be the identity function, therefore, $\mc{T}_s=\{(\GG_i,Y_i)\}_{i \in [s]}$, because when graphs are labeled $\QQ_i(\GG_i)=\GG_i$.
For \emph{shuffled} graph classification $\PP_{\QQ}$ is assumed to be uniform over the permutation matrices, so that all label information is both unavailable and irrecoverable.  The training data are therefore $\mc{T}_s'=\{(\GG_i',Y_i)\}_{i \in [s]}$, where $\GG_i'=\QQ_i(\GG_i)$.  For \emph{unlabeled} graph classification the training data are again $\mc{T}_s'$. 
Our task is to utilize training data to induce a classifier
that approximates 
a Bayes classifier
as closely as possible.



A \emph{labeled} graph Bayes plugin classifier,
$\mh{h}_s: \mc{G}_n \times (\mc{G} \times \mc{Y})^s \to \mc{Y}$,
estimates the parameters $\{\theta_y,\pi_y\}_{y \in \mc{Y}}$ using the training data $\mc{T}_s=\{(\GG_i,Y_i)\}_{i \in [s]}$, and then plugs those estimates into the labeled Bayes classifier, Eq. \eqref{eq:La_Bayes}, resulting in
\begin{align} \label{eq:La_Plug}
	\mh{h}_s(G)=\argmax_{y\in\mc{Y}} \hth_{G|y}\mh{\pi}_y,
\end{align}
where the dependency on the training data is implicit in the $\mh{h}_s(G)$ notation.

A \emph{shuffled} graph Bayes plugin classifier,
$\mh{h}_s': \mc{G}_n \times (\mc{G} \times \mc{Y})^s \to \mc{Y}$,
estimates the parameters $\{\theta_y',\pi_y\}_{y \in \mc{Y}}$ using the training data $\mc{T}_s'=\{(\GG_i',Y_i)\}_{i \in [s]}$, and then plugs those estimates into the shuffled Bayes classifier, Eq. \eqref{eq:Sh_Bayes}, resulting in
\begin{align} \label{eq:Sh_Plug}
	\mh{h}_s'(G)=\argmax_{y\in\mc{Y}} \hth_{G|y}'\mh{\pi}_y.
\end{align}

An \emph{unlabeled} graph Bayes plugin classifier,
$\mh{\mt{h}}_s: \mtc{G}_n \times (\mtc{G}_n \times \mc{Y})^s \to \mc{Y}$,
first determines in which unlabeled set each shuffled graph resides, using $\psi$ as defined in Section \ref{sec:gi}.  Then, it estimates the parameters $\{\mt{\theta}_{\psi(G')|y}\}_{y \in \mc{Y}}$ and $\{\pi_y\}_{y \in \mc{Y}}$ using the training data $\mc{T}_s'$. Finally, it plugs those estimates into the unlabeled Bayes classifier, Eq. \eqref{eq:Un_Bayes}, resulting in
\begin{align} \label{eq:Un_Plug}
	\mh{\mt{h}}_s(G)=\argmax_{y\in\mc{Y}} \mh{\mt{\theta}}_{\mt{G}|y}\mh{\pi}_y.
\end{align}

For brevity, we will sometimes refer to the above three induced classifiers as simply ``classifiers''.  Moreover, the sequence of classifiers (for example, $\{h_s\}_{s \to \infty}$) we will also refer to as a ``classifier''.

\subsection{Bayes Plug-in Graph Classifiers are Universally Consistent} 
\label{sub:bayes_plug_in_classifiers_are_consistent}


The three parametric classifiers, Eqs. \eqref{eq:La_Plug}--\eqref{eq:Un_Plug}, admit classifier estimators that exist, are unique, and moreover, are universally consistent, although the relative convergence rates and values that they converge to differ.

Let $\mh{L}_s=L(\mh{h}_s)$ be the risk of the induced \emph{labeled} graph Bayes plugin classifier using the training data $\mc{T}_s$ to obtain maximum likelihood estimators for $\{\theta_y,\pi_y\}_{y \in \mc{Y}}$. Note that $\mh{L}_s$ is a random variable, as it is a function of the random training data $\mc{T}_s$. This yields
\begin{lem} \label{thm:4}
	$\mh{L}_s \pconv L_*$ as $s \conv \infty$.
\end{lem}
\begin{proof}
Because ${\mc{G}}_n$ and $\mc{Y}$ are both finite, the maximum likelihood estimates for the categorical parameters $\{\theta_y,\pi_y\}_{y \in \mc{Y}}$ are guaranteed to exist and be unique \cite{Devroye1996}.  Hence, the labeled graph Bayes plugin classifier is universally consistent to ${L}_*$ (that is, it converges to ${L}_*$ regardless of the true joint distribution, $\PP_{\QQ, \GG ,Y}$) \cite{Devroye1996}. 
\end{proof}

Similarly, let $\mh{L}_s'=L(\mh{h}_s')$ be the risk of the induced \emph{shuffled} graph Bayes plugin classifier using the training data $\mc{T}_s'$ to obtain maximum likelihood estimators for $\{\theta_y',\pi_y\}_{y \in \mc{Y}}$.  This yields
\begin{coro} \label{cor:Sh_Plug}
	$\mh{L}_s' \pconv L_*'$ as $s \conv \infty$.
\end{coro}
\begin{proof}
	The previous proof rests on the finitude of $\mc{G}_n$, which remains finite after shuffling (uniform or otherwise), and therefore, the previous proof holds, replacing $L_*$ with $L_*'$.
\end{proof}

Thus while one could merely plug the shuffled graphs into $\theta_y'$, such a procedure is inadvisable.  Specifically, the above procedure does not  use the fact that all $\theta_{G_i'|y} = \theta_{G_j'|y}$ whenever $Q(G_i)=G_j$ for some $Q \in \mc{Q}$.  Instead, consider the risk $\mh{\mt{L}}_s=L(\mh{\mt{h}}_s)$ of the induced \emph{unlabeled} graph Bayes plugin classifier upon using the $\psi$ function to map each shuffled graph to its corresponding unlabeled graph, and then obtaining maximum likelihood estimates of the unlabeled graph parameters, $\mt{\theta}$.  
\begin{coro} \label{cor:Un_Plug}
	$\mh{\mt{L}}_s \pconv \mt{L}_*$ as $s \conv \infty$.
\end{coro}

Because $|\mt{\mc{G}}_n| \ll |\mc{G}_n|$ (by a factor of approximately $n!$), it follows that classifying by first projecting the graphs into a lower dimensional space should yield improved performance.  Specifically, we have the following result:

\begin{lem} \label{thm:tdomp}
	$\mh{\mt{h}}_s$ dominates $\mh{h}_s'$ for \emph{shuffled} graph data. 	
\end{lem}
\begin{proof}
Consider the scalar $\mt{\theta}_{\mt{G}|y}$ decomposed into the vector $(\mt{\theta}_{G_1|y},\ldots,\mt{\theta}_{G_{|\mt{G}|}|y})$, where each $\mt{\theta}_{G_i|y}=\mt{\theta}_{\mt{G}|y}/|\mt{G}|$. Note that each $\mt{\theta}_{G_i|y} = \theta'_{G_i|y}$.  Yet, the estimators, $\mh{\mt{\theta}}_{G_i|y}$ and $\mh{\theta}'_{G_i|y}$ are not equal, because the former can borrow strength from all shuffled graphs within the same unlabeled graph, but the latter does not.  Assuming without loss of generality that the class priors are equal and known, the above domination claim is equivalent to stating that for each $G$,
\begin{align}
	\PP[\argmax_{y \in \mc{Y}} \mh{\mt{\theta}}_{G|y} \neq \argmax_{y \in \mc{Y}} \mt{\theta}_{G|y} | \mc{T}_s'] \leq 
	 \PP[\argmax_{y \in \mc{Y}} \mh{\theta}'_{G|y}  \neq \argmax_{y \in \mc{Y}} \theta'_{G|y} | \mc{T}_s'].
\end{align}
Because $\mt{\theta}_{G|y}=\theta'_{G|y}$, the only difference between the two sides of the above inequality is the estimators.  We know that the estimators have the following distributions:
\begin{subequations}
\begin{align}
	s_{\mt{G}} \mh{\mt{\theta}}_{G|y} \sim \text{Binomial}(\mt{\theta}_{G|y}, s_{\mt{G}}) \\
	s_G \mh{\theta}_{G|y} \sim \text{Binomial}(\mt{\theta}_{G|y}, s_{G}),
\end{align}
\end{subequations}
where $s_{\mt{G}}$ is the number of observations of any $G \in \mt{G}$ in the training data, and $s_G$ is the number of observations of $G$ in the training data.  From this, we see that for each $G$, $\mh{\mt{\theta}}_{G|y}$ will have a tighter concentration around the truth due to is borrowing strength, because $s_{\mt{G}} \geq s_G$, so our result holds. 
\end{proof}

\subsection{$k_s$ Nearest Neighbor Graph Classifiers are Universally Consistent} 
\label{sec:a_practical_approach_to_unlabeled_graph_classification}

Corollary \ref{cor:Un_Plug} demonstrates that one can induce a universally consistent classifier $\mh{\mt{h}}_s$  using Eq. \eqref{eq:Un_Plug}. Lemma \ref{thm:tdomp} further shows that the performance of $\mh{\mt{h}}_s$ dominates $\mh{h}_s'$.  Yet, using $\mh{\mt{h}}_s$ is practically useless for two reasons.  First, it requires solving $s$ graph isomorphism problems. Unfortunately, there are no algorithms for solving graph isomorphism problems with worst-case performance known to be in only polynomial time \cite{Fortin1996}. Second, the number of parameters to estimate is super-exponential in $n$ ($\mt{d}_n \approx 2^{n^2}/n!$), and acceptable performance will typically require $s \gg \mt{d}_n$.  We can therefore not even store the parameter estimates for small graphs (e.g., $n=30$), much less estimate them.  
This motivates consideration of an alternative strategy.

A $k_s$ nearest-neighbor ($k$NN) classifier 
using Euclidean norm distance
is universally consistent to $L_*$ for vector-valued data 
as long as $k_s \conv \infty$ with $k_s/s \conv 0$ as $s \conv \infty$ \cite{Stone1977}. This non-parametric approach circumvents the need to estimate many parameters in high-dimensional settings such as graph-classification. The universal consistency proof for $k_s$NN was extended to graph-valued data in reference \cite{VP11_super}, which we include here for completeness.   Specifically, to compare labeled graphs, reference \cite{VP11_super} considered a Frobenius norm distance
\begin{align}
	\delta(G_i,G_j)=\norm{A_i-A_j}_F^2,
\end{align}
where $A_i$ is the adjacency matrix representation of the labeled graph, $G_i$.  
Let $\mh{h}_s^\delta$ denote the Frobenius norm $k_s$NN classifier on \emph{labeled} graphs using $\delta$, and let $\mh{L}^{\delta}_s$ indicate the misclassification rate for this classifier.  Reference \cite{VP11_super} showed:
\begin{lem} \label{thm:5}
	$\mh{L}^{\delta}_s \pconv L_*$ as $s \conv \infty$.
\end{lem}
\begin{proof}
Because both $\mc{G}$ and $\mc{Y}$ have finite cardinality, the law of large numbers ensures that eventually as $s \conv \infty$, the plurality of nearest neighbors to a test graph will be identical to the test graph. 
\end{proof}
Let $\mh{h}_s^{'\delta}$ denote the Frobenius norm $k_s$NN classifier on \emph{shuffled} graphs using $\delta'$, and let $\mh{L}^{'\delta'}_s$ indicate the misclassification rate for this classifier.   
From the above lemma and Corollary \ref{cor:Sh_Plug},  the below follows immediately:
\begin{coro} \label{cor:knn1}
	$\mh{L}^{'\delta}_s \pconv L_*'$ as $s \conv \infty$.
\end{coro}
Given shuffled graph data $\mc{T}_s'$, however, other distance metrics appear more ``natural'' to us.  For example,  
consider the ``graph-matched Frobenius norm'' distance:
\begin{align} \label{eq:d}
\delta'(G_i',G_j')=\min_{Q \in \mc{Q}_n}\norm{Q(A_i')-A_j'}_F^2,	
\end{align}
where $A_i'$ and $A_j'$ are shuffled adjacency matrices.  
Let $\mh{h}^{'\delta'}_s$ indicate the misclassification rate of the $k_s$NN classifier using the above graph-matched norm $\delta'$ \emph{shuffled} graphs, and let $\mh{L}^{'\delta'}_s$ indicate the misclassification rate for this classifier.  Given an exact graph matching function---a function that actually solves Eq. \eqref{eq:d}---we have the following result:
\begin{coro} \label{cor:Sh_knn}
	$\mh{L}^{'\delta'}_s \pconv L_*'$ as $s \conv \infty$.
\end{coro}
Thus, given shuffled data $\mc{T}_s'$, one could consider either $\mh{h}_s^{\delta}$ or $\mh{h}_s^{'\delta'}$.  

Interestingly, when the data are labeled graphs, $\mc{T}_s$, one can outperform $\mh{h}_s^{\delta}$ by \emph{shuffling}, that is, by apparently destroying the label information.  Consider an example in which $\theta=\theta'$, such that no information is in the labels.  In such scenarios, shuffling can effectively borrow strength from different labeled graphs that are within the same unlabeled graph set. Let $\mh{h}^{\delta'}_s$ indicate the misclassification rate of the $k_s$NN classifier using $\delta'$ \emph{labeled} graphs, and let $\mh{L}^{\delta'}_s$ indicate the misclassification rate for this classifier. We therefore state without proof:
\begin{lem} \label{thm:nodom}
	Neither $\mh{h}_s^{\delta'}$ nor $\mh{h}_s^{\delta}$ dominates when data are \emph{labeled} graphs.
\end{lem}

Thus, when the training data consists of shuffled graphs, the best universally consistent classifier (of those considered herein) is a $k_s$NN that uses $\delta'$ as the distance metric.  Other universally consistent classifiers that we considered either require estimating more parameters than there are molecules in the universe, or are inadmissible under $0-1$ loss.  When vertex labels are available, no classifier dominates.

\subsection{Comparing Asymptotic Performances} 
\label{sub:asymptotically_optimal_classifiers}


The above theoretical results consider Bayes plug-in and $k_s$NN classifiers. Here we consider other classifiers.  Specifically,  let $\mh{L}^{\psi}_s$ be the misclassification rate for some classifier that operates on $\mc{T}_s'$, that is, only has access to shuffled graphs.  Consider the set of seven graph invariants studied in \cite{PCP10}: 
size, max degree, max eigenvalue, scan statistic, number of triangles, and average path length. %
Via Monte Carlo, \cite{PCP10} was unable to find a uniformly most powerful graph invariant (test statistic \cite{priebe2010you}) for a particular hypothesis testing scenario with unlabeled graphs.  The above results, however, indicate that there exists optimal classifiers (or test statistics) for any unlabeled or shuffled graph setting.  
To proceed, let $\mh{h}^{\mh{\pi}}_s$ be the \emph{chance} classifier, that is
\begin{align}
	\mh{h}_s^{\mh{\pi}}(G)=\argmax_{y\in\mc{Y}} \mh{\pi}_y,
\end{align}
and let $\mh{L}^{\mh{\pi}}_s$ be the misclassification rate for this classifier. Moreover, let $\mh{L}_s^\psi$ be the risk of the invariant classifier that is equivalent to the unlabeled Bayes plug-in classifier (see Lemma \ref{thm:3}). From the above results, it follows that:
\begin{lem} \label{thm:order}
	In expectation, \\ $\mh{L}^{\mh{\pi}}_s \geq  \mh{L}_s'=\mh{\mt{L}}_s = \mh{L}^{\psi}_s = \mh{L}^{\delta'}_s=\mh{L}^{'\delta'}_s$ as $s \conv \infty$. 
\end{lem}

\subsection{Comparing Computational Properties} 
\label{sub:comparing_computational_properties}

While asymptotic results can be informative and insightful, understanding the computational properties of the different classifiers can be as (or even more) informative for real applications.  Table \ref{tab:comp} compares the space and time complexity of the various classifiers considered above.  Only the $k$NN classifiers have the property that they do not require more space than there are atoms in the universe (for any $n$ bigger than $\approx 30$). Of those, the labeled $k_s$NN classifier does not require time exponential in the number of vertices. Therefore, we only found one type of classifier with performance guarantees that has both polynomial space and time.  Unfortunately, the finite sample performance of this classifier is abysmal.  This motivates constructing approximate classifiers. 

\begin{table}[h!]
\caption{Order of computational properties for training the various shuffled graph classifiers.}
\begin{center}
\begin{tabular}{|r|c|c|c|}
\hline 
name & notation &  time  &  space \\ \hline 
chance & $\mh{h}^{\mh{\pi}}_s$ & s & 1 \\ 
invariant & $\mh{h}^{\psi}_s$  & $e^n s$ & $\min(2^{\binom{n}{2}}/n!,s)$ \\ 
labeled Bayes plug-in & $\mh{h}_s'$ &  $s$ & $\min(2^{\binom{n}{2}},s)$  \\ 
unlabeled Bayes plug-in & $\mht{h}_s $ & $e^n s$ & $\min(2^{\binom{n}{2}}/n!,s)$ \\ 
labeled $k_s$NN & $\mh{h}^{\delta'}_s$  & $n^2 s^2$ & $n^2 s$ \\ 
unshuffled $k_s$NN & $\mh{h}^{'\delta'}_s$ & $ e^n s^2$ & $n^2 s$ \\ 
    \hline
\end{tabular}
\end{center}
\label{tab:comp}
\end{table}%


\section{Real World Application} 
\label{sec:simulated_experiment}

We buttress the above theoretical results via numerical experiments.  The asymptotic results combined with the computational complexities of the above described algorithm suggest that none of the proposed algorithms have all the properties we effectively require for real world applications, in particular, polynomial space and time complexity, as well as reasonable convergence rates.  We therefore propose a different algorithm, which lacks universal consistency, but can be run on real data with good hope for reasonable performance.  In particular, we modify $\mh{h}_s^{'\delta'}$, the \emph{unshuffled} $k$NN classifier.  Instead of requiring this classifier to actually solve the graph matching problem, Eq. \eqref{eq:d}, we  use a recently proposed state-of-the-art approximate cubic time algorithm \cite{VP11_QAP}.  Denote this classifier $\mh{h}_s^{'\mt{\delta}}$.

\subsection{Shuffled Connectome Classification} 
\label{sub:connectome_classification}

A ``connectome'' is a brain-graph in which vertices correspond to (groups of) neurons, and edges correspond to connections between them.  Diffusion Magnetic Resonance (MR) Imaging and related technologies are making the acquisition of MR connectomes routine \cite{Hagmann2010}.  49 subjects from the Baltimore Longitudinal Study on Aging comprise this data, with acquisition and connectome inference details as reported in \cite{MRCAP11}.  Each connectome yields a $70$ vertex simple graph (binary, symmetric, and hollow adjacency matrix).  Associated with each graph is class label based on the sex of the individual (24 males, 25 females).  Because the vertices are labeled, we can compare the results of having the labels and not having the labels.

Consider the following five classifiers:
\begin{itemize}
	\item {$\delta$-$1$NN:}  A $1$-nearest neighbor ($1$NN) with Frobenius norm distance on the \emph{labeled} adjacency matrices.
	\item {$\delta'$-$1$NN:} A $1$NN with Frobenius norm distance on the \emph{shuffled} adjacency matrices.
	\item {$\mt{\delta}$-$1$NN:} A $1$NN with an \emph{approximate} graph-matched Frobenius norm distance on the shuffled adjacency matrices, as described above.  Because graph-matching is $\mc{NP}$-hard \cite{Garey1979}, we instead use an inexact graph matching approach based on the quadratic assignment formulation described in \cite{VP11_QAP}, which only requires $\mc{O}(n^3)$ time.
	\item {$\psi$-$1$NN:} A $1$NN with Euclidean distance using the seven graph invariants described above.  Prior to computing the Euclidean distance, for each invariant, we rescale all the values to lie between zero and one.
	\item {$\mh{\pi}$:} Use the chance classifier defined above.
\end{itemize}
Performance is assessed by leave-one-out misclassification rate.

Figure \ref{fig:1} reifies the above theoretical results in a particular finite sample regime.  We apply the five algorithms discussed above to sub-samples of the connectome data, which shows approximate convergence rates for this data.  Fortunately, this real data example supports the main lemmas of this work.  Specifically, the $k_s$NN classifier using $\delta$ on the \emph{labeled} graphs (dashed gray line) achieves the lowest misclassification rate for all $s$, which one would expect if labels contain appropriate class signal.  
Moreover, 
the $k_s$NN classifier using the inexact graph-matching Frobenius norm on the shuffled adjacency matrices, $\mt{\delta}$, performs best of all classifiers using only shuffled graphs (compare dashed black line with solid black and gray lines).  On the other hand, while the $k_s$NN classifier using the Frobenius norm on shuffled graphs, $\delta'$, must eventually converge to $L_s'$, its convergence rate is quite slow, so the classifier using standard invariants $\psi$ outperforms the simple $\delta'$ based $k_s$NN.

\begin{figure}[htbp]
	\centering
		\includegraphics[width=0.5\linewidth]{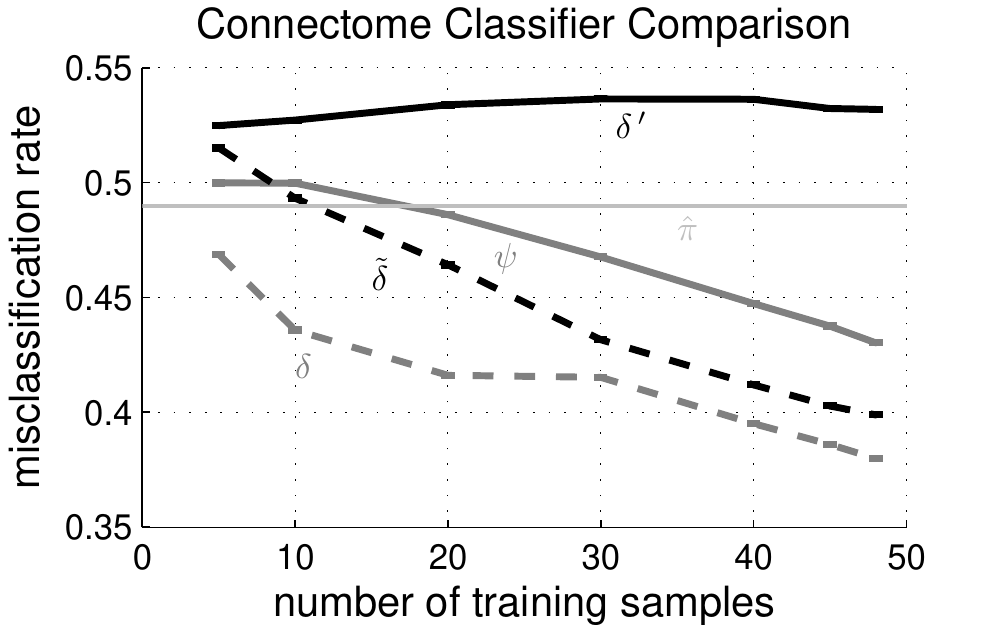}
		\caption{Connectome misclassification rates for various classifiers.  2000 Monte Carlo sub-samples of the data were performed for each $s$, such that errorbars were neglibly small.  Five classifiers were compared, as described in main text.  Note that when $s$ is larger than $20$, as predicted by theory, we have $\mh{L}^{\mh{\pi}}_s > \mh{L}^\psi_s > \mh{L}^{\mt{\delta}}_s > \mh{L}_s^\delta$. Moreover, $\mh{L}^{\delta'}_s > \mh{L}^{\mt{\delta}}_s > \mh{L}_s^\delta$.}
	\label{fig:1}
\end{figure}


\section{Discussion}

In this work, we address both the theoretical and practical limitations of classifying shuffled graphs, relative to labeled and unlabeled graphs.  Specifically, first we construct the notion of shuffled graphs and shuffled graph classifiers in a parallel fashion with labeled and unlabeled graphs/classifiers, as we were unable to find such notions in the literature.  Then, we show that shuffling the vertex labels results in an irretrievable situation, with a possible degradation of classification performance (Lemma \ref{thm:1}). Even if the vertex labels contained class-conditional signal, Bayes performance may remain unchanged (Lemma \ref{thm:2}).  Moreover, although one cannot recover the vertex labels, one can obtain a Bayes optimal classifier by solving a large number of graph isomorphism problems (Lemma \ref{thm:3}).  This resolves a theoretical conundrum: is there a set of graph invariants that can yield a universally consistent graph classifier?  When the generative distribution is unavailable, one can induce a consistent and efficient ``unshuffling'' classifier by using a graph-matching strategy (Corollary \ref{cor:Un_Plug}).  While this unshuffling approach dominates the more na\"ive approach (Lemma \ref{thm:tdomp}), it is intractable in practice due to the difficulty of graph matching and the large number of isomorphism sets.  Instead, a Frobenius norm $k_s$NN classifier applied to the adjacency matrices may be used, which is also universally consistent (Corollary \ref{cor:Sh_knn}).  
Surprisingly, none of the considered classifiers dominate the other for labeled data (Lemma \ref{thm:nodom}), yet asymptotically, we can order shuffled graph classifiers (Lemma \ref{thm:order}).

Because graph-matching is $\mc{NP}$-hard, we instead use an approximate graph-matching algorithm in practice (see \cite{VP11_QAP} for details).  Applying these $k_s$NN classifiers to a problem of considerable scientific interest---classifying human MR connectomes---we find that even with a relatively small sample size ($s \geq 20$), the approximately graph-matched $k_s$NN algorithm performs nearly as well as the $k_s$NN algorithm \emph{using} vertex labels, and slightly better than a $k_s$NN algorithm applied to a set of graph invariants proposed previously \cite{PCP10}.  This suggests that the asymptotics might apply even for very small sample sizes.  Thus, this theoretical insight has led us to improved practical classification performance.  Extensions to weighted or (certain) attributed graphs are straightforward.

\section*{Acknowledgment}

This work was partially supported by the Research Program in Applied Neuroscience.


\bibliographystyle{IEEEtran}

\end{document}